\RequirePackage[l2tabu, orthodox]{nag}

\documentclass[orivec]{llncs}

\usepackage{fixltx2e}[2006/09/13]

\usepackage{lmodern}
\usepackage[utf8]{inputenc}
\usepackage[T1]{fontenc}
\usepackage{microtype}

\usepackage{amsmath}
\usepackage{graphicx}
\usepackage{tabu}

\usepackage[ruled,linesnumbered,noend]{algorithm2e}

\usepackage[noadjust]{cite}
\usepackage{placeins}
\usepackage{float}

\usepackage{enumitem}

\title{Improved Purely Additive\\ Fault-Tolerant Spanners\thanks{This work was
partially supported by the Research Grant PRIN 2010 ``ARS TechnoMedia", funded by the Italian Ministry of Education, University, and Research, and by the ERC Starting Grant ``New Approaches to Network Design".}}
\author{Davide Bil\`o\inst{1} \and Fabrizio Grandoni\inst{2}
\and Luciano Gual\`a\inst{3} \and\\ Stefano Leucci\inst{4} \and Guido Proietti\inst{4,5}
\institute{Dipartimento di Scienze Umanistiche e Sociali,
Università di Sassari, Italy \and IDSIA, University of Lugano, Switzerland \and Dipartimento di Ingegneria dell'Impresa,
Università di Roma ``Tor Vergata", Italy \and DISIM, Università degli Studi dell'Aquila, Italy  \and Istituto di Analisi dei Sistemi
  ed Informatica,
  CNR, Roma, Italy \\
E-mail: \texttt{davide.bilo@uniss.it; fabrizio@idsia.ch; guala@mat.uniroma2.it; stefano.leucci@univaq.it; guido.proietti@univaq.it}
}}

\usepackage{xspace}
\newcommand{\col}{\texttt{color}\xspace}
\newcommand{\cnt}{\texttt{counter}\xspace}
\newcommand{\white}{\texttt{white}\xspace}
\newcommand{\black}{\texttt{black}\xspace}
\newcommand{\red}{\texttt{red}\xspace}

\newcommand{\wdeg}{\ensuremath{\delta_\white}}
\newcommand{\wneigh}{\ensuremath{N_\white}}

\newcommand{\softO}{\widetilde{O}}

\let\doendproof\endproof
\renewcommand\endproof{~\hfill\qed\doendproof}

\newcommand{\hide}[1]{\relax}

\usepackage{color}

\ifdefined\DEBUG

\definecolor{darkgreen}{RGB}{0,125,32}

\newcommand{\fab}[1]{\textcolor{red}{#1}}
 \newcommand{\ste}[1]{\textcolor{blue}{#1}}

  \def\rem#1{{\marginpar{\raggedright\scriptsize #1}}}
  \newcommand{\fabr}[1]{\rem{\textcolor{red}{$\bullet$ #1}}}
  \newcommand{\ster}[1]{\rem{\textcolor{blue}{$\bullet$ #1}}}
  \newcommand{\guir}[1]{\rem{\textcolor{darkgreen}{$\bullet$ #1}}}
  \newcommand{\lucr}[1]{\rem{\textcolor{cyan}{$\bullet$ #1}}}
  \newcommand{\davr}[1]{\rem{\textcolor{magenta}{$\bullet$ #1}}}

\else

  \newcommand{\fab}[1]{#1}
  \newcommand{\ste}[1]{#1}

  \newcommand{\fabr}[1]{}
  \newcommand{\ster}[1]{}
  \newcommand{\guir}[1]{}
  \newcommand{\lucr}[1]{}
  \newcommand{\davr}[1]{}

\fi

\begin{document}

\pagestyle{plain}	
\maketitle

\begin{abstract}
Let $G$ be an unweighted $n$-node undirected graph. A \emph{$\beta$-additive spanner} of $G$ is a spanning subgraph $H$ of $G$ such that distances in $H$ are stretched at most by an additive term $\beta$ w.r.t. the corresponding distances in $G$. A natural research goal related with spanners is that of designing \emph{sparse} spanners with \emph{low} stretch. 

In this paper, we focus on \emph{fault-tolerant} additive spanners, namely additive spanners which are able to preserve their additive stretch even when one edge fails. We are able to improve all known such spanners, in terms of either sparsity or stretch. In particular, we consider the sparsest known spanners with stretch $6$, $28$, and $38$, and reduce the stretch to $4$, $10$, and $14$, respectively (while keeping the same sparsity). 

Our results are based on two different constructions. On one hand, we show how to augment (by adding a \emph{small} number of edges) a fault-tolerant additive \emph{sourcewise spanner} (that approximately preserves distances only from a given set of source nodes) into one such spanner that preserves all pairwise distances. On the other hand, we show how to augment some known fault-tolerant additive spanners, based on clustering techniques. This way we decrease the additive stretch without any asymptotic increase in their size.
We also obtain improved fault-tolerant additive spanners for the case of one vertex failure, and for the case of $f$ edge failures.
\end{abstract}

\section{Introduction}
We are given an unweighted, undirected $n$-node graph $G=(V(G),E(G))$. Let $d_G(s,t)$ denote the shortest path distance between nodes $s$ and $t$ in $G$. A \emph{spanner} $H$ of $G$ is a spanning subgraph such that $d_H(s,t) \le \varphi(d_G(s,t))$ for all $s$, $t \in V(G)$, where $\varphi$ is the so-called \emph{stretch} or \emph{distortion} function of the spanner.
In particular, when $\varphi(x)=\alpha x + \beta$, for constants $\alpha, \beta$, the spanner is named an ($\alpha$,$\beta$) spanner. If $\alpha=1$, the spanner is called \emph{(purely) additive} or also $\beta$\emph{-additive}. If $\beta=0$, the spanner is called $\alpha$-\emph{multiplicative}.

Finding \emph{sparse} (i.e., with a small number of edges) spanners is a key task in many network applications, since they allow for a small-size infrastructure onto which an efficient (in terms of paths' length) point-to-point communication can be performed.
Due to this important feature, spanners were the subject of an intensive research effort, aiming at designing increasingly sparser spanners with lower stretch. 

However, as any sparse structure, a spanner is very sensitive to possible failures of \emph{components} (i.e., edges or nodes), which may drastically affect its performances, or even disconnect it! Thus, to deal with this drawback, a more robust concept of \emph{fault-tolerant} spanner is naturally conceivable, in which the distortion must be guaranteed even after a subset of components of $G$ fails.

More formally, for a subset $F$ of edges (resp., vertices) of $G$, let $G-F$ be the graph obtained by removing from $G$ the edges (resp., vertices and incident edges) in $F$. When $F=\{x\}$, we will simply write $G-x$. Then, an \emph{$f$-edge fault-tolerant} ($f$-EFT) spanner with distortion $(\alpha, \beta)$, is a subgraph $H$ of $G$ such that, for every set $F \subseteq E(G)$ of at most $f$ failed edges, we have\footnote{Note that in this definition we allow $d_{G-F}(s,t)$ to become infinite (if the removal of $F$ disconnects $s$ from $t$). In that case we assume the inequality to be trivially satisfied.}
\[
	d_{H-F}(s, t) \le \alpha \cdot d_{G-F}(s,t) + \beta \quad \forall s,t \in V(G).
\]
We define similarly an \emph{$f$-vertex fault-tolerant} ($f$-VFT) spanner. For $f=1$, we simply call the spanner edge/vertex fault-tolerant (EFT/VFT).

Chechik et al. \cite{CLPR09} show how to construct a $(2k-1)$-multiplicative $f$-EFT spanner of size $O(f \cdot n^{1+ 1/k})$, for any integer $k\geq 1$. Their approach also works for weighted graphs and for vertex-failures, returning a $(2k-1)$-multiplicative $f$-VFT spanner of size $\softO(f^2 \cdot k^{f+1} \cdot n^{1+1/k})$.\footnote{The $\softO$ notation hides poly-logarithmic factors in $n$.} This latter result has been finally improved through a randomized construction in \cite{DK11}, where the expected size was reduced to $\softO(f^{2-1/k} \cdot n^{1+1/k})$.
For a comparison, the sparsest known $(2k-1)$-multiplicative \emph{standard} (non fault-tolerant) spanners have size $O(n^{1+\frac{1}{k}})$ \cite{DBLP:journals/dcg/AlthoferDDJS93}, and this is believed to be asymptotically tight due to the girth conjecture of Erd\H{o}s \cite{erdHos1964extremal}.

Additive fault-tolerant spanners can be constructed with the following approach by Braunshvig et al \cite{BCP12}. Let $M$ be an $\alpha$-multiplicative $f$-EFT spanner, and $A$ be a $\beta$-additive standard spanner. Then $H=M\cup A$ is a $(2f(2\beta+\alpha-1)+\beta)$-additive $f$-EFT spanner. One can exploit this approach to construct \emph{concrete} EFT spanners as follows.  We know how to construct $6$-additive spanners of size $O(n^{4/3})$ \cite{BKMP10}, randomized spanners that, w.h.p., have size $\softO(n^{7/5})$ and additive distortion $4$ \cite{C13}, and $2$-additive spanners of size $O(n^{3/2})$ \cite{ACIM99}. By setting $f=1$ and choosing $k$ properly, this leads to EFT spanners of size $O(n^{4/3})$ with additive distortion $38$, size $\softO(n^{7/5})$ with additive distortion $28$ (w.h.p.), and size $O(n^{3/2})$ with additive distortion $14$. Finally, using a different approach, Parter \cite{P14} recently presented 2- and 6-additive EFT/VFT spanners of size $\softO(n^{5/3})$ and $\softO(n^{3/2})$, respectively.

\subsection{Our Results.}

In this paper, we focus on additive EFT spanners, and we improve all the known such spanners in terms of sparsity or stretch (see Table \ref{table:new results}). We also present some better results for additive VFT and $f$-EFT spanners.

\begin{table}[t]
\setlength{\tabulinesep}{0.8mm} 
\setlength{\tabcolsep}{1.5mm} 
\centering
\caption{State of the art and new results on additive EFT spanners. Distortions and sizes marked with ``*'' hold w.h.p.}
\begin{tabu}{|c|c|c|c|}
	\hline
	\multicolumn{2}{|c|}{State of the art} & \multicolumn{2}{c|}{Our results}\\ \hline
	Size & {\large$\substack{\text{Additive}\\\text{distortion}}$} & Size & {\large$\substack{\text{Additive}\\\text{distortion}}$} \\ \hline
	$\softO(n^{5/3})$ & $2$ \cite{P14}	&$O(n^{5/3})$&$ 2$\\ \hline	
	$\softO(n^{3/2})$ & $6$ \cite{P14} &$O(n^{3/2})$&$4$\\ \hline
	$\softO(n^\frac{7}{5})$* & $28$* \cite{BCP12,C13} & $\softO(n^\frac{7}{5})$* & $10$*\\ \hline
	$O(n^\frac{4}{3})$ & $38$ \cite{BCP12,BKMP10} & $O(n^\frac{4}{3})$ & $14$\\ \hline
\end{tabu}
\label{table:new results}
\end{table}

In more detail, our improved EFT spanners exploit the following two novel approaches. Our first technique (see Section \ref{sec:spanner_algorithm}), assumes that we are given an additive \emph{sourcewise} fault-tolerant spanner $A_S$, i.e., a fault-tolerant spanner that guarantees low distortion only for the distances from a given set $S$ of source nodes. We show that, by carefully choosing $S$ and by augmenting $A_S$ with a conveniently selected \emph{small} subset of edges, it is possible to construct a fault-tolerant spanner (approximately preserving \emph{all} pairwise distances) with a moderate increase of the stretch. This, combined with the sourcewise EFT spanners in \cite{BGLP14,PP13}, leads to the first two results in the table. In particular, we reduce the additive stretch of the best-known spanner of size $\tilde{O}(n^{3/2})$ from $6$ \cite{P14} to $4$ (actually, we also save a polylogarithmic factor in the size here). For the case of stretch $2$, we slightly decrease the size from $\softO(n^{5/3})$ \cite{P14} to $O(n^{5/3})$.
This technique also applies to VFT spanners. In particular, we achieve a $2$-additive VFT spanner of size $O(n^{5/3})$ rather than $\softO(n^{5/3})$ \cite{P14}, and a $4$-additive VFT spanner of size $O(n^{3/2}\sqrt{\log n})$, improving on the $6$-additive  VFT spanner of size $\softO(n^{3/2})$ in \cite{P14}.

Our second technique (see Section \ref{sec:augmenting_clustering_spanners}) relies on some properties of known additive spanners. We observe that some known additive spanners are based on \emph{clustering techniques} that construct a small-enough number of clusters. Furthermore, the worst-case stretch of these spanners is achieved only in some specific cases. We exploit these facts to augment the spanner $H=M\cup A$ based on the already mentioned construction of \cite{BCP12} with a small number of inter and intra-cluster edges. This allows us to reduce the additive stretch without any asymptotic increase in the number of edges.

Finally, for the case of multiple edge failures, we are able to prove that the construction in \cite{BCP12} has in fact an additive stretch of only $2f(\beta + \alpha - 1) + \beta$ (rather than $2f(2\beta + \alpha-1) + \beta$).
\begin{theorem}\label{thm:multifault}
Let $A$ be a $\beta$-additive spanner of $G$, and let $M$ be an $\alpha$-multiplicative $f$-EFT spanner of $G$. The graph $H=(V(G), E(A) \cup E(M))$ is a $(2f(\beta + \alpha - 1) + \beta)$-additive $f$-EFT spanner of $G$. In the special case $f=1$, the additive stretch is at most $2\beta+\alpha-1$.
\end{theorem}
We see this as an interesting result since, to the best of our knowledge, the construction of \cite{BCP12} is the only known approach for building additive spanners withstanding more than a single edge fault. At a high-level, in \cite{BCP12} the shortest path in $G-F$ between two vertices is decomposed into (roughly $2f$) subpaths called \emph{blocks}. The authors then show that it is possible to build a \emph{bypass} (i.e., a fault-free path) in $H$ between the endpoints of each block such that the additive error incurred by using this path is at most $\beta + \alpha - 1$.
Actually, in addition to those intra-block bypasses, the spanner $H$ contains some inter-block \emph{shortcuts}, that are exploited in order to prove a better distortion. The proof of this result is given in the appendix. 

\subsection{Related Work}

A notion closely relate to fault-tolerant spanners is the one of 
\emph{Distance Sensitivity Oracles} (DSO). The goal here is to compute, with a \emph{low} preprocessing time, a \emph{compact} data structure which is able to  \emph{quickly} answer distance queries following some component failures (possibly in an approximate way). 
For recent achievements on DSO, we refer the reader to \cite{BK09,CLPR10,BK13,GW12}. 

Another setting which is very close in spirit to fault-tolerant spanners is the recent work on fault-tolerant \emph{approximate shortest-path trees}, both for unweighted~\cite{PP14} and for weighted  \cite{BK13,BGLP14} graphs. In \cite{AFIR13} it was introduced the resembling concept of \emph{resilient spanners}, i.e., spanners that approximately preserve the \emph{relative} increase of distances due to an edge failure.

There was also some research (see for example \cite{CE06,CGK13}) on spanners approximately preserving the distance from a given set of nodes (\emph{sourcewise} spanners), among a given set of nodes (\emph{subsetwise} spanners), or between given pairs of nodes (\emph{pairwise} spanners). In this framework a spanner is called a \emph{preserver} if distances are preserved (in other words, the stretch function is the identity function). In particular, in one of our constructions we exploit a fault-tolerant version of a sourcewise preserver.

\subsection{Notation}

Given an unweighted, undirected graph $G$, let us denote by $\pi_G(u,v)$ a shortest path) between $u$ and $v$ in $G$. When the graph $G$ is clear from the context we might omit the subscript. 
Given a simple path $\pi$ in $G$ and two vertices $s,t \in V(\pi)$, we define $\pi[s,t]$ to be the subpath of $\pi$ connecting $s$ and $t$. Moreover, we denote by $|\pi|$ the \emph{length} of $\pi$, i.e., the number of its edges.
When dealing with one or multiple \emph{failed edges}, we say that a path $\pi$ is \emph{fault-free} if it does not contain any of such edges.
Finally, if two paths $\pi$ and $\pi^\prime$ are such that the last vertex of $\pi$ coincides with the first vertex of $\pi^\prime$, we will denote by $\pi \circ \pi^\prime$ the path obtained by concatenating $\pi$ with $\pi^\prime$.

\section{Augmenting Sourcewise Fault-Tolerant Spanners}
\label{sec:spanner_algorithm}

We next describe a general procedure (see Algorithm \ref{alg:b+2_ft_spanner}) to derive purely additive fault-tolerant spanners from sourcewise spanners of the same type.

The main idea of the algorithm is to select a small subset $S$ of \emph{source} vertices of $G$, which we call \red. These vertices are used to build a fault-tolerant sourcewise spanner of the graph.
The remaining vertices are either \black or \white. The former ones are always adjacent to a source, even in the event of $f$ edge/vertex failures. Finally, edges incident to \white vertices are added to the sought spanner, as their overall number is provably small.

During the execution of the algorithm, we let $\col(u) \in \{ \white, \black, \red \}$ denote the current color of vertex $u$.
We define $\wneigh(u)$ to be the set of neighbors of $u$ which are colored \white, and we let $\wdeg(u)=|\wneigh(u)|$.
We will also assign a non-negative counter $\cnt(u)$ to each vertex. Initially all these counters will be positive, and then they will only be decremented.  A vertex $u$ is colored $\black$ only when  $\cnt(u)$ reaches $0$, and once a \white vertex is colored either \black or \red it will never be recolored \white again. Therefore, we have that $\col(u)=\black$ implies $\cnt(u)=0$.

\begin{algorithm}[t]
	\DontPrintSemicolon
		
	$\col(v) \gets \white \quad \forall v \in V$; $\cnt(v) \gets f+1 \quad \forall v \in V$\;
	$S \gets \emptyset$; $E' \gets \emptyset$ \;
	
	\BlankLine	
	
	\While{$\exists s \in V \setminus S : \wdeg(s) \ge p$}{
		$S \gets S \cup \{ s \}$ \tcc*{Add a new source $s$}
		$\col(s) \gets \red$ \;
		\ForEach{$u \in \wneigh(s)$}{
			$\cnt(u) \gets \cnt(u)-1$ \;
			$E' \gets E' \cup \{ (s, u) \}$ \label{ln:dominating_edges} \;
			\If{$\cnt(u)=0$}{
				$\col(u) \gets \black$ \;
			}
		}
	}

	\BlankLine
	
	$E' \gets E' \cup \{ (u,v) \in E : \col(u)=\white \}$ \label{ln:white_edges} \;
	$A_S \gets \text{$\beta$-additive f-EFT/VFT sourcewise spanner w.r.t. sources in $S$}$ \;
	\Return $H \gets (V(G),E'\cup E(A_S))$

	\caption{Algorithm for computing a fault-tolerant additive spanner of $G$ from a $\beta$-additive $f$-EFT/VFT sourcewise spanner. The parameter $p$ affects the size of the returned spanner and it will be suitably chosen.}\label{alg:b+2_ft_spanner}
\end{algorithm}

We first bound the size of the spanner $H$.
\begin{lemma}\label{lem:sourcewise:size}
Algorithm \ref{alg:b+2_ft_spanner} computes a spanner of size $O\left(np + nf + \gamma\left(n, \left\lfloor \frac{(f+1)n}{p} \right\rfloor \right) \right)$, where $\gamma(n,\ell)$ is the size of the spanner $A_S$ for $|S|=\ell$.
\end{lemma}
\begin{proof}

The edges added to $E'$ by line~\ref{ln:dominating_edges} are at most $(f+1)n$, since each time an edge $(s,u)$ is added to $H$ the counter $\cnt(u)$ is decremented, and at most $(f+1)n$ counter decrements can occur. 

To bound the edges added to $E'$ by line \ref{ln:white_edges}, observe that all the edges in $\{ (u,v) \in E : \col(u)=\white \}$ which are incident to a \red vertex, have already been added to $H$ by line \ref{ln:dominating_edges}, hence we only consider vertices $v$ which are either \white or \black. Let $v$ be such a vertex and notice that, before line \ref{ln:white_edges} is executed, we must have $\wdeg(v) < p$, as otherwise $v$ would have been selected as a source and colored \red. This immediately implies that line \ref{ln:white_edges} causes the addition of  at most $np$ edges to $H$.

It remains to bound the size of $A_S$. It is sufficient to show that $|S| \le \left\lfloor \frac{(f+1)n}{p} \right\rfloor$ at the end of the algorithm. 
Each time a source $s$ is selected, $s$ has at least $p$ white neighbors in $G$, hence the quantity $\sum_{u \in V(G)} \cnt(u)$ decreases by at least $p$.
	The claim follows by noticing that $\sum_{u \in V(G)} \cnt(u)=n(f+1)$ at the beginning of the algorithm and, when the algorithm terminates, it must be non-negative.
\end{proof}

\begin{figure}[!tb]
	\centering
	\includegraphics[scale=1.2]{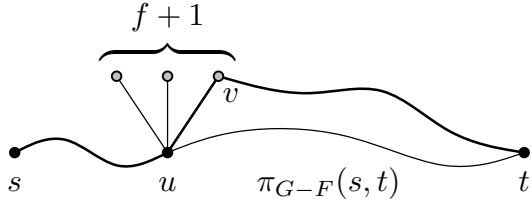}
	\caption{A case of the proof of Theorem~\ref{thr:beta_2_spanner}. Bold edges are in $H-F$. The black vertex $u$ is adjacent to at least $f+1$ vertices in $S$.}\label{fig:spanner-analysis}
\end{figure}

We next bound the distortion of $H$. 	
\begin{theorem}
	\label{thr:beta_2_spanner}
	Algorithm \ref{alg:b+2_ft_spanner} computes a $(\beta+2)$-additive $f$-EFT/VFT spanner. 
\end{theorem}
\begin{proof}
	Consider two vertices $s,t \in V(G)$ and a set $F$ of at most $f$ failed edges/vertices of $G$, we will show that $d_{H-F}(s,t) \le d_{G-F}(s,t) + \beta + 2$. We assume, w.l.o.g., that $s$ and $t$ are connected in $G-F$, as otherwise the claim trivially holds.
	
	If all the vertices in $\pi = \pi_{G-F}(s,t)$ are \white, then all their incident edges have been added to $H$ (see line \ref{ln:white_edges} of Algorithm~\ref{alg:b+2_ft_spanner}), hence $d_{H-F}(s,t) = d_{G-F}(s,t)$.
	
	Otherwise, let $u \in V(\pi)$ be the closest vertex to $s$ such that $\col(u) \neq \white$. Notice that, by the choice of $u$, $H-F$ contains all the edges of $\pi[s, u]$. If $\col(u) = \red$ then:
	\begin{align*}
		d_{H - F}(s,t) & \le d_{H - F}(s,u) + d_{H - F}(u,t) \le  d_{G - F}(s,u) + d_{A_S - F}(u,t) \\
		& \le d_{G - F}(s,u) + d_{G - F}(u,t) + \beta = d_{G-F}(s,t) + \beta
	\end{align*}
	\noindent where we used the fact that $u \in \pi_{G-F}(s,t)$.

	Finally, if	$\col(u) = \black$ then $\cnt(u)=0$, hence $u$ has at least $f+1$ \red neighbors in $H$ (see Figure~\ref{fig:spanner-analysis}). As a consequence, there is at least one \red vertex $v$ such that $(u,v) \in H - F$ (and hence $(u,v) \in G - F$), therefore:
	\begin{align*}
		d_{H - F}(s,t) & \le d_{H - F}(s,u) + d_{H - F}(u,v) + d_{H - F}(v,t) \\
	& 	\le d_{G - F}(s,u) + 1 + d_{A_S - F}(v,t) \le d_{G - F}(s,u) + 1 + d_{G - F}(v,t) + \beta  \\
	&	\le d_{G - F}(s,u) + 1 + d_{G - F}(v,u) + d_{G - F}(u,t) + \beta\\
	& =  d_{G-F}(s,t) + \beta + 2.
	\end{align*}
\end{proof}

Let $S^\prime \subset V(G)$ be a set of \emph{sources}. In \cite{PP13} it is shown that a sourcewise EFT/VFT \emph{preserver} (i.e, a $(1,0)$ EFT/VFT sourcewise spanner) of $G$ having size $\gamma(n, |S^\prime|) = O( n \sqrt{n|S^\prime|})$ can be built in polynomial time. Combining this preserver and Algorithm~\ref{alg:b+2_ft_spanner} with $p=n^\frac{2}{3}$, we obtain the following:
\begin{corollary}
There exists a polynomial time algorithm to compute a $2$-additive EFT/VFT spanner of size $O(n^{\frac{5}{3}})$.
\end{corollary}

Furthermore, we can exploit the following result in \cite{BGLP14}.\footnote{Actually, the result in \cite{BGLP14} is claimed for the single source case only, but it immediately extends to multiple sources.}
\begin{lemma}[\cite{BGLP14}]\label{lem:GBLP14}
Given an $(\alpha, \beta)$-spanner $A$ and a subset $S^\prime$ of vertices, it is possible to compute in polynomial time a subset of $O(|S^\prime| \cdot n)$ edges $E'$, so that $A\cup E'$ is an $(\alpha, \beta)$ EFT sourcewise spanner w.r.t. $S^\prime$. The same result holds for VFT spanners, with $E'$ of size $O(|S^\prime| \cdot n \log n )$.
\end{lemma}

Combining the above result with the $2$-additive spanner of size $O(n^{3/2})$ in~\cite{ACIM99}, we obtain $2$-additive EFT and VFT sourcewise spanners of size $\gamma(n, |S^\prime|) = O(n \sqrt{n} + |S^\prime| \cdot n)$ and $\gamma(n, |S^\prime|) = O(n \sqrt{n} + |S^\prime| \cdot n \log n)$, respectively. By using these spanners in Algorithm~\ref{alg:b+2_ft_spanner}, with $p=\sqrt{n}$ and $p=\sqrt{n \log n}$ respectively, we obtain the following result. 

\begin{corollary}
There exists a polynomial time algorithm to compute a $4$-additive EFT spanner of size $O(n^{3/2})$, and a $4$-additive VFT spanner of size $O(n^{3/2}\sqrt{\log n})$.
\end{corollary}

\section{Augmenting Clustering-Based Additive Spanners}
\label{sec:augmenting_clustering_spanners}

\newcommand{\C}{\mathcal{C}}
\newcommand{\GC}{G_\C}
\newcommand{\ccenter}{\mbox{cnt}}
\newcommand{\tpi}{\tilde\pi}

Most additive spanners in the literature are based on a \emph{clustering} technique. A subset of the vertices of the graph $G$ is partitioned into  \emph{clusters}, each containing a special \emph{center} vertex along with some of its neighbors. The distances between these clusters is then reduced by adding a suitable set of edges to the spanner. This technique is used, for example, in \cite{BKMP10,C13}.
We now describe a general technique which can be used to augment such spanners in order to obtain a fault-tolerant additive spanner.

More formally, a \emph{clustering} of $G$ is a partition $\C$ of a subset of $V(G)$. We call each element $C \in \C$ a \emph{cluster}. We say that a vertex $v$ is \emph{clustered} if it belongs to a cluster, and \emph{unclustered} otherwise.
Each cluster $C \in \C$ is associated with a vertex $u \in C$ which is the \emph{center} of $C$.
For each \emph{clustered} vertex $v$, we denote by $\ccenter(v)$ the center of the cluster containing $v$.

We say that a $\beta$-additive spanner $A$ is \emph{clustering-based} if there exists a clustering $\C$ of $G$ such that: (i) $A$ contains all the edges incident to unclustered vertices, (ii) $A$ contains all the edges between every clustered vertex $v$ and $\ccenter(v)$, and (iii) the following property holds:

\begin{property}
	\label{property:1_6_spanner_sp_structure}
	For every $u,v \in V(G)$ such that $v$ is a clustered vertex, there exists a path $\tpi(u,v)$ in $A$ such that one of the following conditions holds: 
	\begin{enumerate}[label=(P\arabic*)]
		\item $|\tpi(u,v)| \le d_G(u,v) + \beta - 2$; \label{item:1_6_spanner_sp_structure_rule_1}
		\item \label{item:1_6_spanner_sp_structure_rule_2} $|\tpi(u,v)| = d_G(u,v) + \beta -1$ and either (i) $v=\ccenter(v)$, or (ii) the last edge of $\tpi(u,v)$ is $(\ccenter(v),v)$.
		\item \label{item:1_6_spanner_sp_structure_rule_3} $|\tpi(u,v)| = d_G(u,v) + \beta$, $v \neq \ccenter(v)$, and the last edge of $\tpi(u,v)$ is $(\ccenter(v),v)$.
	\end{enumerate}
\end{property}

Our algorithm works as follows (see Algorithm \ref{alg:augmenting_clusters}). We add to our spanner $H$ a $\beta$-additive clustering-based spanner $A$, and a $\alpha$-multiplicative EFT spanner $M$. Note that so far our construction is the same as in \cite{BCP12}, with the extra constraint that $A$ is clustering-based. We then augment $H$ by adding a carefully chosen subset $E'$ of inter and intra-cluster edges.

\begin{algorithm}[t]
	\DontPrintSemicolon
	$E' \gets \emptyset$; $M \gets (\alpha, 0)$ EFT spanner; $A\gets (1,\beta)$ clustering-based spanner \;
	
	\ForEach{$C \in \C$\label{ln:augmenting_intra_cluster}}{
		\ForEach{$v \in C$}{
			\If{$\exists (v,x) \in E(G) : x \in C \setminus \{ \ccenter(v) \}$}{
				$E' \gets E' \cup \{ (v,x) \}$.
			}
		}
	}

	\BlankLine
	
	\ForEach{$C,C^\prime \in \C : C \neq C^\prime$ \label{ln:augmenting_inter_cluster}}{
				\If{$\exists e,e^\prime \in \delta(C, C^\prime) : e$ and $e^\prime$ are vertex-disjoint}{
					$E' \gets E' \cup \{ e, e^\prime \}$.
				}
				\ElseIf{$\exists e,e^\prime \in \delta(C, C^\prime) : e \neq e^\prime $}{
					$E' \gets E' \cup \{ e, e^\prime \}$.				
				}
				\Else {
					$E' \gets E' \cup \delta(C, C^\prime)$ \tcc*{$\delta(C, C^\prime)$ contains at most one edge}
				}
	}

	\BlankLine	
	
	\Return $H \gets (V(G), E' \cup E(M) \cup E(A))$

	\caption{Algorithm for computing a fault-tolerant additive spanner from multiplicative and clustering-based additive spanners. Here $\{\C,\ccenter(\cdot)\}$ denotes the clustering of $G$ while $\delta(C, C^\prime)$ is the set of the edges in $E(G)$ with one endpoint in $C$ and the other in $C^\prime$.}
	\label{alg:augmenting_clusters}
\end{algorithm}

Let $\{\C,\ccenter(\cdot)\}$ be the clustering of $A$. It is easy to see that $E'$ contains at most $O(n + |\C|^2)$ edges and hence $|E(H)| = O(|E(A)| + |E(M)| + |\mathcal{C}|^2)$. We now prove an useful lemma which is then used to upper-bound the distortion of the spanner $H$.

\begin{figure}[t]
	\centering
	\includegraphics[scale=1]{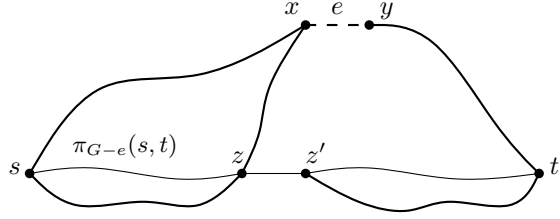}
	\caption{Decomposition of $\pi_{G-e}(s,t)$ so that all shortest paths from $s$ to $z$ (resp. from $z^\prime$ to $t$) in $A$ are fault-free. Bold lines denote shortest paths in $A$.}\label{fig:spanner_union_one_fault}
\end{figure}

\begin{lemma}
	\label{lemma:spanner_union_one_fault_structure}
	Let $A$ be a spanning subgraph of $G$, let $e \in E(G)$ be a failed edge and $s,t \in V(G)$ be two vertices satisfying $d_{A}(s,t) < d_{A-e}(s,t) \neq \infty$. There exist two consecutive vertices $z,z^\prime$ in $V(\pi_{G-e}(s,t))$, with $d_{G-e}(s,z) < d_{G-e}(s,z^\prime)$, such that every shortest path in $A$ between $s$ and $z$ (resp. $t$ and $z^\prime$) is fault-free.
\end{lemma}
\begin{proof}
	
	First of all, notice that $e=(x,y)$ belongs to every shortest path between $s$ and $t$ in $A$, therefore let $\pi_A(s,t) = \langle s, \dots, x, y, \dots, t \rangle$.
	Consider the vertices of $\pi_{G-e}(s,t)$ from $s$ to $t$, let $z \in V(\pi)$ be the last vertex such that there exists a shortest path $\pi$ between $z$ and $t$ in $A$ that contains $e$ ($z$ can possibly coincide with $s$), and call $z^\prime$ the vertex following $z$ in $\pi$ (see Figure~\ref{fig:spanner_union_one_fault}). 	
	By the choice of $z$, we have that no shortest path between $z^\prime$ and $t$ in $A$ can contain $e$.
	Moreover, $\pi$ must traverse $e$ in the same direction as $\pi_A(s,t)$, i.e., $\pi = \langle z, \dots, x, y, \dots, t \rangle$. This is true since otherwise we would have $\pi = \langle z, \dots, y, x, \dots, t \rangle$ and hence $\pi_A(s,t)[s,x] \circ \pi[x,t]$ would be a fault-free shortest path between $s$ and $t$ in $A$, a contradiction.
	
	It remains to show that no shortest path between $s$ and $z$ in $A$ can contain $e$. Suppose this is not the case, then:
	\[
		d_A(s,z) = d_A(s,y) + d_A(y,z) = d_A(s,x) + 1 + 1 + d_A(z, x) = d_A(s, z) + 2 
	\]
	which is again a contradiction.
\end{proof}

\noindent We are now ready to prove the main theorem of this section.

	\begin{theorem}
		\label{thm:clustering_eft_spanner}
		Algorithm \ref{alg:augmenting_clusters} computes a $(2\beta + \max\{2, \alpha-3\})$-additive EFT-spanner.	
	\end{theorem}		
	\begin{proof}
		Choose any two vertices $s,t \in V(G)$ and a failed edge $e \in E(G)$.
		Suppose that $s,t$ are connected in $G-e$, and that every shortest path between $s$ and $t$ in $A$ contains $e$ (as otherwise the claim trivially holds). We partition $\pi_{G-e}(s,t)$ by finding $z,z^\prime \in V(\pi_{G-e}(s,t))$ as shown by Lemma~\ref{lemma:spanner_union_one_fault_structure}.

	\begin{figure}[tb]
		\centering
		\includegraphics[width=\textwidth]{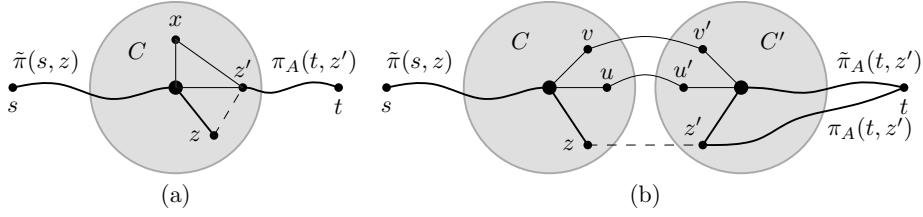}
		\caption{Cases considered in the proof of Theorem~\ref{thm:clustering_eft_spanner} to build a fault-free path between $s$ and $t$ with small additive distortion. Bold lines represent shortest paths/edges in $A$. Solid lines represent paths/edges in $H$ while the dashed edge $(z, z^\prime)$ does not belong to $E(H)$ and cannot coincide with $e$.}
		\label{fig:clusttering_eft_sanner}
	\end{figure}

		The edge $(z, z^\prime)$ is in $\pi_{G-e}(s,t)$ and hence cannot coincide with $e$. Moreover, we suppose $(z, z^\prime) \not\in E(H)$, as otherwise we would immediately have:
\begin{align*}		
d_{H-e}(s,t) & \le d_{A}(s,z) + d_{H-e}(z,z^\prime) + d_A(z^\prime, t) \\
& \le d_{G-e}(s,z) + \beta + 1 + d_{G-e}(z^\prime,t) + \beta + ( d_{G-e}(z, z^\prime) - 1 )\\
& = d_{G-e}(s,t) + 2\beta - 1. 
\end{align*}
This means that both $z$ and $z^\prime$ must be clustered. Let $C$ (resp. $C^\prime$) be the (unique) cluster that contains $z$ (resp. $z^\prime$).\footnote{Notice that $C$ and $C^\prime$ may coincide.}
		
Let $\gamma:=d_A(s,z) - d_{G-e}(s,z)$ and $\gamma^\prime:= d_A(t, z^\prime) - d_{G-e}(z^\prime, t)$. Clearly $0\leq \gamma,\gamma^\prime\leq \beta$.		
If $\gamma + \gamma^\prime \le 2\beta-2$ then we are done as:
\begin{align*}
d_{H-e}(s,t) & \le d_{A}(s,z) + d_{B-e}(z,z^\prime) + d_A(z^\prime, t)\\ & \le d_{G-e}(s, z) + d_{G-e}(z^\prime, t) + 2\beta - 2 + \alpha d_{G-e}(z,z^\prime) + (d_{G-e}(z,z^\prime) - 1)\\ 
& \le d_{G-e}(s,t) + 2\beta + \alpha - 3.
\end{align*}
		
		Next we assume that $\gamma + \gamma^\prime \ge 2\beta -1$. This means that either (i) $\gamma$ and $\gamma^\prime$ are both equal to $\beta$, or (ii) exactly one of them is $\beta$ while the other equals $\beta - 1$. Assume w.l.o.g.\ that $\gamma = \beta$. This implies that $d_A(s, z) = |\tpi(s,z)|$, hence $\tpi(s,z)$ is a shortest path between $s$ and $z$ in $A$, and by Lemma~\ref{lemma:spanner_union_one_fault_structure}, it is fault-free.
		
	In the rest of the proof we separately consider the cases $C=C^\prime$ and $C \neq C^\prime$. In the former case, since $(z, z^\prime) \not\in E(H)$, we know that, during the execution of the loop in line \ref{ln:augmenting_intra_cluster} of Algorithm~\ref{alg:augmenting_clusters}, an edge $(z^\prime,x)$ such that $x \in C \setminus \{z, \ccenter(z) \}$ has been added to $E$ (see Figure~\ref{fig:clusttering_eft_sanner} (a)).
 Since the paths $\langle \ccenter(z), z^\prime \rangle$ and $\langle \ccenter(z), x, z^\prime \rangle$ are edge-disjoint, at least one of them is fault free, hence: $d_{H-e}(\ccenter(z), z^\prime) \le 2 = d_{G-e}(z, z^\prime) + 1$. Thus, by \ref{item:1_6_spanner_sp_structure_rule_3} of Property~\ref{property:1_6_spanner_sp_structure}:
\begin{align*} 
 d_{H-e}(s, t) & \le d_A(s, \ccenter(z)) + d_{H-e}(\ccenter(z), z^\prime) + d_A(z^\prime, t)\\
 & \le d_{G-e}(s, z) + \beta - 1 + d_{G-e}(z, z^\prime) + 1 + d_{G-e}(z^\prime, t) + \beta  \\
 & \le d_{G-e}(s,t) + 2\beta.
 \end{align*}
	
	We now consider the remaining case, namely $C \neq C^\prime$. We have that $(z, z^\prime) \not\in E(H)$, therefore during the execution of the loop in line \ref{ln:augmenting_inter_cluster} of Algoritm~\ref{alg:augmenting_clusters}, two distinct edges $(u,u^\prime), (v, v^\prime)$ so that $u,v \in C$ and $u^\prime, v^\prime \in C^\prime$ must have been added to $E$ (see Figure~\ref{fig:clusttering_eft_sanner} (b)).

	Notice that $u^\prime$ and $v^\prime$ might coincide, but this would imply that $u \neq v$ and hence $u^\prime = v^\prime = z^\prime$. This, in turn, implies the existence of two edge-disjoint paths of length $2$ between $\ccenter(z)$ and $z^\prime$ in $H$, namely $\langle \ccenter(z), u, z^\prime \rangle$ and $\langle \ccenter(z), v, z^\prime \rangle$. As at least one of them must be fault-free. Therefore: 
\begin{align*}
d_{H-e}(s,t) &  \le d_{A}(s, \ccenter(z)) + d_{H-e}(\ccenter(z), z^\prime) + d_{A}(z^\prime,t) \\
 & \le d_{G-e}(s, z) + \beta - 1 + d_{G-e}(z, z^\prime) + 1 + d_{G-e}(z^\prime,t) + \beta \\
 & = d_{G-e}(s,t) + 2\beta.
\end{align*}

	On the other hand, if $u^\prime \neq v^\prime$, we consider the two paths $\pi^\prime = \langle \ccenter(z), u, u^\prime, \ccenter(z^\prime) \rangle$ and $\pi^{\prime\prime} = \langle \ccenter(z), v, v^\prime, \ccenter(z^\prime) \rangle$.\footnote{Some consecutive vertices of $\pi^\prime$ (resp. $\pi^{\prime\prime}$) might actually coincide. In this case, we ignore all but the first of such vertices and define $\pi^\prime$ (resp. $\pi^{\prime\prime}$) accordingly.}
	Notice that $\pi^\prime$ and $\pi^{\prime\prime}$ can share at most a single edge, namely $(\ccenter(z), z)$ (when $u=v=z$), and that this edge cannot coincide with $e$ as it belongs to $\tpi(s,z)$ which is a fault-free shortest path between $s$ and $z$ in $A$. This implies that at least one of $\pi^\prime$ and $\pi^{\prime\prime}$ is fault-free and hence $d_{H-e}(\ccenter(z), \ccenter(z^\prime)) \le 3 = d_{G-e}(z, z^\prime) + 2$.
	If $e = (\ccenter(z^\prime), z^\prime)$ then, since $|\tpi(t,z^\prime)| \ge d_A(z^\prime, t) \ge d_{G-e}(z^\prime, t)+ \beta - 1$, either \ref{item:1_6_spanner_sp_structure_rule_2} or \ref{item:1_6_spanner_sp_structure_rule_3} of Property~\ref{property:1_6_spanner_sp_structure} must hold, so we know that $\tpi(t,z^\prime)[t,\ccenter(z^\prime)]$ is fault-free and has a length of at most $d_{G-e}(z^\prime, t) + \beta - 1$. We have:
	\begin{align*}	
	d_{H-e}(s, t)  & \le  d_{A}(s, \ccenter(z)) + d_{H-e}(\ccenter(z), \ccenter(z^\prime)) + d_A(\ccenter(z^\prime),t)\\
	 & \le d_{G-e}(s, z) + \beta - 1 + d_{G-e}(z, z^\prime) + 2 + d_{G-e}(z^\prime, t) + \beta -1 \\
	 & \le d_{G-e}(s,t) + 2\beta.
\end{align*}	
Finally, when $e \neq (\ccenter(z^\prime), z^\prime)$, we have:
\begin{align*}
	d_{H-e}(s,t) & \le d_A(s,\ccenter(z)) + d_{H-e}(\ccenter(z), \ccenter(z^\prime)) + d_{H-e}(\ccenter(z^\prime), z^\prime) + d_{A}(z^\prime, t) \\
	 & \le d_{G-e}(s, z) + \beta -1 + d_{G-e}(z, z^\prime) + 2 + 1 + d_{G-e}(z^\prime, t) + \beta \\
	 &  \le d_G(s,t) + 2\beta + 2.
\end{align*}
This concludes the proof.
\end{proof}

This result can immediately be applied to the $6$-additive spanner of size $O(n^\frac{4}{3})$ in \cite{BKMP10}, which is clustering-based and uses $O(n^{\frac{2}{3}})$ clusters. 
Using the $5$-multiplicative EFT spanner $M$ of size $O(n^{4/3})$ from \cite{CLPR09}, we obtain:
\begin{corollary}
There exists a polynomial time algorithm to compute a $14$-additive EFT spanner of size $O(n^{\frac{4}{3}})$.
\end{corollary}

We can similarly exploit the clustering-based spanner of \cite{C13} which provides, w.h.p., an additive stretch of $4$ and a size of $\softO(n^\frac{7}{5})$ by using $O(n^\frac{3}{5})$ clusters.

\begin{corollary}
	There exists a polynomial time randomized algorithm that computes w.h.p. a $10$-additive EFT spanner of $G$ of size $\softO(n^{\frac{7}{5}})$. 
\end{corollary}

\bibliographystyle{plain}
\bibliography{bibliography}

\newpage

\appendix

\section{Purely-Additive Fault Tolerant Spanners for $f$ Faults}
\label{sec:spanner_union_multiple_faults}

\newcommand{\class}{C}

As mentioned earlier, the authors of \cite{BCP12} construct an additive $f$-EFT spanner $H$ by merging a $\beta$-additive spanner $A$ with an $\alpha$-multiplicative $f$-EFT spanner $M$. They show that the additive distortion of their spanner is (at most) 
$2f(2\beta + \alpha - 1) + \beta$. We show that indeed their spanner has a lower distortion $2f(\beta + \alpha - 1) + \beta$.

Choose any two (distinct) vertices $s,t \in V(G)$ \ste{which are connected in $G-F$}, and consider the path $\pi=\pi_{G-F}(s,t)$.
Given two distinct vertices $u,v \in V(\pi)$ we say that either $u < v$, $u = v$, or $u > v$ if $d_A(s,u) < d_A(s,v)$, $d_A(s,u) = d_A(s,v)$, or $d_A(s,u) > d_A(s,v)$, respectively.
Moreover, for each pair of vertices $u,v \in V(G)$ we define their \emph{class} $\class(u,v)$ to be the first failed edge $(x,y)$ that is encountered by traversing $\pi_A(u,v)$ from $u$ to $v$. We consider $\class(u,v)$ as a \emph{directed} edge so that if $\class(u,v)=(x,y)$ then $d_A(u,x)<d_A(u, y)$ and $\class(u,v) \neq (y,x)$.
If the path $\pi_A(u,v)$ is fault-free then we let $\class(u,v) = \Phi$.

Now, following \cite{BCP12}, we partition $\pi$ into a sequence $\langle B_1, B_2, \dots, B_k, B^* \rangle$ of consecutive edge-disjoint subpaths of $\pi$ which we call \emph{blocks}. Each block $B \neq B^*$ is associated to five special vertices, namely $x^B_1$, $x_B^2$, $y^B_1$, $y^B_2$ and $y^B_\bot$. In particular $B$ starts with $x^B_1$ and ends with $y^B_\bot$. When $B$ is clear from the context we might omit the superscript.

The first block $B=B_0$ of $\pi_A(s,t)$ is defined as follows: suppose that $\class(s,t) \neq \Phi$ and let $x_1=s$ and $x_2$ be the first (i.e., the \emph{smallest}) vertex of $\pi$ such that $\class(x_1,x_2) = f \neq \Phi$. Then, let $(y_1, y_2)$ be the last \fab{pair} of vertices of $\pi$ (in lexicographic order) such that $\class(y_1, y_2)=f$ (notice that $(y_1,y_2)$ might coincide with $(x_1,x_2)$). We define $y_\bot$ to be the vertex following $y_1$ in $\pi$. The block $B$ is the subpath of $\pi$ starting in $x_1$ and ending in $y_\bot$. Notice that (i) $x_1 < x_2$, (ii) $x_1 \le y_1 < y_\bot$, and (iii) $V(B)$ might not contain $x_2$ and/or $y_2$.
The next block $B_1$ is defined in a similar manner, on the remaining portion of the path $\pi$, namely $\pi[y^B_\bot, t]$. We continue this recursive procedure until a block $B_k$ such that $\class(y^{B_k}_\bot, t) = \Phi$ is found, when this happens we set $x^{B^*}_1=y^{B_k}_\bot$, $y^{B^*}_\bot = t$, and $B^* = \pi[y^{B_k}_\bot, t]$, hence completing the partition. Notice that $B^*$ might be the only block of the decomposition if $\class(s,t) = \Phi$, and that $B^*$ might also be a path of length $0$ (which contains the sole vertex $t$). Moreover, notice that if $B$ and $B^\prime$ are two consecutive blocks, then $y^B_\bot = x^{B^\prime}_1$.

For the sake of convenience, we extend the definition of $\class$ 
to blocks in the following way: $\class(B) = \class(x^B_1, y^B_\bot)$. 
Moreover, given two blocks $B, B^\prime$, we say that $B^\prime > B$ if $x^{B^\prime}_1 > x^B_1$.

It is shown in \cite{BCP12} that, for each block, there exists a bypass of small additive distortion between its endpoints. 
\begin{lemma}[\cite{BCP12}]
	\label{lemma:single_block_shortcut}
	For each block $B$, we have $d_{H-F}(x^B_1, y^B_\bot) \le d_{G-F}(x^B_1, y^B_\bot) + 2\beta + \alpha - 1$.
\end{lemma}

Let us prove a few technical lemmas.

\begin{lemma}
	\label{lemma:one_time_fault}
	For each block $B \neq B^*$, and for each couple of vertices $u,v \in V(\pi)$ such that $y^B_\bot \le u \le v$, we have $\class(u,v) \neq \class(B)$.
\end{lemma}
\begin{proof}
It follows immediately from how $\pi$ is partitioned into blocks.
\end{proof}

\begin{lemma}
	\label{lemma:single_block_fast_shortcut}
	For each block $B \neq B^*$ such that $\class(x^B_1, y^B_1) = \Phi$ we have: \linebreak $d_{H-F}(x^B_1, y^B_\bot) \le d_{G-F}(x^B_1, y^B_\bot) + \beta + \alpha -1$.
\end{lemma}
\begin{proof}
	From the fact that $(y_1, y_\bot) \in E(\pi)$, we know that $(y_1, y_\bot) \not\in F$. Therefore:
\begin{align*}	 
d_{H-F}(x_1,y_\bot) & \le d_{A-F}(x_1,y_1) + d_{M-F}(y_1, y_\bot) = d_{A}(x_1,y_1) + \alpha d_{G-F}(y_1, y_\bot)\\
& \le d_{G}(x_1, y_1) + ( d_{G-F}(y_1, y_\bot) - 1 )+ \beta + \alpha \\ 
& \le d_{G-F}(x_1, y_1) + \beta + \alpha - 1.
\end{align*}
\end{proof}

\begin{figure}[t]
	\centering
	\includegraphics[width=\textwidth]{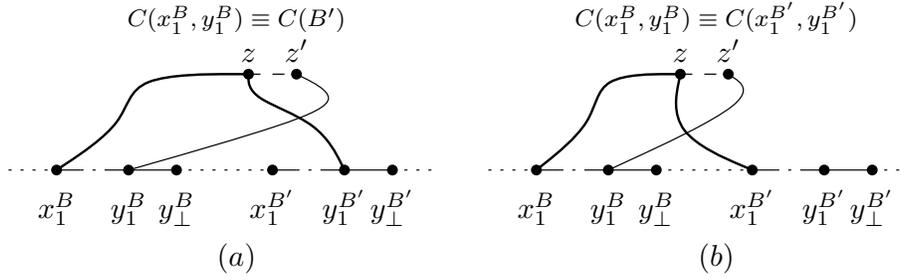}
	\caption{Paths of low additive distortions between successive blocks. The class $C(x^B_1, y^B_1)$ is $(z, z^\prime)$. Bold curves represent fault-free shortest paths in $A$.}\label{fig:spanner_union}
\end{figure}

We now prove \fab{two} useful lemmas that show the existence of \emph{fault-tolerant} paths of low additive distortion between \fab{distinct} blocks:\fabr{Spostato questa frase dopo il lemma 6, che non parla di blocchi distinti. Giusto?}
\begin{lemma}
	\label{lemma:path_block_shortcut}
	For each block $B \neq B^*$ such that (i) $\class(x^B_1, y^B_1) \neq \Phi$, and (ii) there exists a block $B^\prime > B$ satisfying $\class(B^\prime) = \class(x^B_1, y^B_1)$, we have: $d_{H-F}(x^B_1, y^{B^\prime}_\bot) \le d_{G-F}(x^B_1, y^{B^\prime}_\bot) + 2\beta + \alpha -1$.
\end{lemma}
\begin{proof}
	Let $\class(B^\prime) = (z, z^\prime)$. We first prove that $d_{A-F}(x^B_1,z) + d_{A-F}(z, y^{B^\prime}_1 \le d_{G-F}(x^B_1, y^{B^\prime}_1) + 2 \beta$. Indeed we must have $x^B_1 < y_B < y^{B^\prime}_1$, so we can write the following two inequalities (see Figure~\ref{fig:spanner_union} (a)): 	
	\begin{equation}
		d_{A-F}(x^B_1, z) + d_A(z, z^\prime) + d_{A-F}(z^\prime, y^B_1) = d_A(x^B_1, y^B_1) \le d_{G}(x^B_1, y^B_1) + \beta,
		\label{eq:path_block_shortcut_1}
	\end{equation} %
	\begin{multline}
		d_{A-F}(z, y^{B^\prime}_1) = d_{A}(z, y^{B^\prime}_1) \le d_G(z, y^{B^\prime}_1) + \beta \\
		 \le d_A(z, z^\prime) + d_{A-F}(z^\prime, y^B_1) + d_G(y^B_1, y^{B^\prime}_1) + \beta.
		 \label{eq:path_block_shortcut_2}
	\end{multline}

	By summing \eqref{eq:path_block_shortcut_1} and \eqref{eq:path_block_shortcut_2}, we obtain:
	\[
		d_{A-F}(x^B_1, z) + d_{A-F}(z, y^{B^\prime}_1) \le d_{G}(x^B_1, y^B_1) + d_G(y^B_1, y^{B^\prime}_1) + 2\beta \le d_{G-F}(x^B_1, y^{B^\prime}_1) + 2\beta.
	\]

	Remember that $d_{M-F}(y^{B^\prime}_1, y^{B^\prime}_\bot) \le \alpha d_{G-F}(y^{B^\prime}_1, y^{B^\prime}_\bot) = \alpha$. We are now ready to prove the claim:
\begin{align*}
		d_{H-F}(x^B_1,y^{B^\prime}_\bot) & \le d_{A-F}(x^B_1,z) + d_{A-F}(z, y^{B^\prime}_1 ) + d_{M-F}(y^{B^\prime}_1, y^{B^\prime}_\bot)\\
		& \le d_{G-F}(x^B_1, y^{B^\prime}_1) + 2\beta + \alpha + ( d_{G-F}(y^{B^\prime}_1, y^{B^\prime}_\bot) - 1 )\\
		& \le d_{G-F}(x^B_1,y^{B^\prime}_\bot) + 2\beta + \alpha - 1.
\end{align*}
\end{proof}

\begin{lemma}
	\label{lemma:path_path_shortcut}
	For each block $B \neq B^*$ such that (i) $\class(x^B_1, y^B_1) \neq \Phi$, and (ii) there exists a block $B^\prime > B$ with $\class(x^{B^\prime}_1, y^{B^\prime}_1) = \class(x^B_1, y^B_1)$, we have: $d_{H-F}(x^B_1, x^{B^\prime}_1) \le d_{G-F}(x^B_1, x^{B^\prime}_1) + 2 \beta$.
\end{lemma}
\begin{proof}
	Let $\class(B^\prime) = (z, z^\prime)$. We can write the following inequalities (see Figure~\ref{fig:spanner_union} (b)):
	\begin{equation}
		d_{A-F}(x^B_1, z) + d_{A}(z, z^\prime) + d_{A-F}(z^\prime, y^B_1) = d_A(x^B_1, y^B_1) \le d_G(x^B_1, y^B_1) + \beta,
		\label{eq:path_path_shortcut_1}
	\end{equation}
	\begin{align}
		d_{A-F}(z, x^{B^\prime}_1) & = d_{A}(z, x^{B^\prime}_1) \le d_G(z, x^{B^\prime}_1)\nonumber \\
		 & \le d_A(z, z^\prime) + d_{A-F}(z^\prime, y^B_1) + d_G(y^B_1, x^{B^\prime}_1) + \beta.
		\label{eq:path_path_shortcut_2}
	\end{align}
	
	By summing \eqref{eq:path_path_shortcut_1} and \eqref{eq:path_path_shortcut_2} we obtain: $d_{A-F}(x^B_1, z) + d_{A-F}(z, x^{B^\prime}_1)  \le d_G(x^B_1, y^B_1) + d_G(y^B_1, x^{B^\prime}_1) + 2\beta $ which implies $d_{A-F}(x^B_1, x^{B^\prime}_1) \le d_{G-F}(x^B_1, x^{B^\prime}_1) + 2\beta$, and hence the claim.
\end{proof}

We are now ready to prove the main result of this section:
\begin{lemma}\label{lem:multifault}
Let $A$ be a $\beta$-additive spanner of $G$, and let \fab{$M$} be an $\alpha$-multiplicative $f$-EFT spanner of $G$. The graph $H=(V(G), E(A) \cup E(\fab{M}))$ is a $(2\beta+\alpha-1)$-additive $f$-EFT spanner of $G$. 
\end{lemma}
\begin{proof}
Let $s,t\in V(G)$ be any two vertices, \fab{$F$ be any subset of at most $f$ edges}, and $\pi = \pi_{G-F}(s,t)$ \fab{(that we assume to exist)}. Consider the \fab{above} decomposition of $\pi$ into blocks $\langle B_1, B_2, \dots, B_k, B^* \rangle$.

We define a block $B$ to be \emph{traversed} by a subpath $\pi^\prime$ of $\pi$ if $E(\pi^\prime) \cap E(B) \neq \emptyset$.
	A \emph{shortcut} is a fault-tolerant path in $H$ between the vertex $x^B_1$ of a block $B$ and the vertex $x_1^{B^\prime}$ of a subsequent block $B^\prime > B$. We will construct a path from $s$ to $t$ in $H-F$ by building a shortcut from the first block of $\pi$ to another block $B^\prime$. Then, we restrict our path to $\pi[x^{B^\prime}_1, t]$ and repeat the above until we reach $x^{B^*}_1$.  We say that the shortcut from $x^B_1$ to $x^{B^\prime}_1$ \emph{skips} the blocks that traverse the path $\pi[x^B_1, x^{B^\prime}_1]$. Each block $B \neq B^*$ is skipped exactly once.

	We say that a shortcut is \emph{bad} if it is built as shown by Lemma~\ref{lemma:single_block_shortcut} or by Lemma~\ref{lemma:path_path_shortcut}, and \emph{good} otherwise. Intuitively, a bad shortcut involves a class that does not correspond to the class of any block of $\pi$.
	
	Let $B$ the first traversed block of the path $\pi^\prime$ we are considering (initially $\pi^\prime = \pi$ and hence $B=B_1$). We find the sought shortcut in the following way:
	
	If $\class(x^B_1, y^B_1) = \Phi$, then we reach the vertex $x^{B^\prime}_1$ of the following block $B^\prime$ with a (good) shortcut as shown by Lemma~\ref{lemma:single_block_fast_shortcut} (remember that $y^B_\bot = x^{B^\prime}_1$). This shortcut has an additive error of at most $\beta + \alpha - 1$ w.r.t.\ $d_{G-F}(x^B_1, x^{B^\prime}_1)$, and skips $1$ block.
		
	 Otherwise, we search for a block $B^\prime \ge B$ such that $\class(x^B_1, y^B_1) = \class(B^\prime)$. If such a block is found, we build a (good) shortcut from $x^B_1$ to $y^{B^\prime}_\bot$ ah shown by Lemma~\ref{lemma:path_block_shortcut}. This shortcut has an additive error of at most $2\beta + \alpha - 1$ w.r.t.\ $d_{G-F}(x^B_1, y^{B^\prime}_1)$, \fab{and it skips} at least $2$ blocks as we must necessarily have $B^\prime \neq B$ (since otherwise $\class(x^B_1, y^B_1) = \class(x^B_1, x^B_2) \fab{=} \class(y^B_1, y^B_2)$ which contradicts the definition of $B$).
	
	If no such block can be found, then we search for the last block $B^\prime > B$ such that $\class(x^B_1, y^B_1) = \class(x^{B^\prime}_1, y^{B^\prime}_1)$. If $B^\prime$ exists, then we build a  (bad) shortcut from $x^B_1$ to $x^{B^\prime_1}$ as shown by Lemma~\ref{lemma:path_path_shortcut}. This shortcut has an additive error of at most $2\beta$ w.r.t\ $d_{G-F}(x^B_1, x^{B^\prime}_1)$, and skips at least $1$ block.
	
	Finally, if none of the previous cases apply, we build a (bad) shortcut between $x^B_1$ and $y^B_\bot$ as shown by Lemma~\ref{lemma:single_block_shortcut}. This shortcut has an additive error of at most $2\beta + \alpha - 1$ w.r.t\ $d_{G-F}(x^B_1, y^B_\bot)$, and skips $1$ block.
	
	We define $\lambda$ to be the total number of bad shortcuts built for the whole path $\pi$.
	Notice that, by Lemma~\ref{lemma:one_time_fault}, $k$ must be at most $2f$. Moreover, each time we build a bad shortcut from a block $B$ we know that no block of $\pi$ has class $\class(B)$ and that we will not encounter the edge of $\class(B)$ in any following path between two vertices $x_1$ and $y_1$.
	Therefore $\lambda \le 2f - k$. Clearly the number of good shortcuts is at most $k-\lambda$.

	For each block skipped by a good shortcut we incur an additive error of at most $\beta + \alpha -1$, while for each block skipped by a bad shortcut we incur an additive error of at most $2\beta + \alpha - 1$.
	Hence, we have that $d_{H-F}(s, x^{B^*}_1)  \le d_{G-F}(s, x^{B^*}_1) + \gamma$ where 
\begin{align*}
\gamma & \le \max_{0 \le \lambda \le 2f - k} \left( \lambda (2\beta + \alpha -1) + (k-\lambda)(\beta+\alpha-1)  \right)\\
& = \max_{0 \le \lambda \le 2f - k} \left( \lambda \beta + k (\beta + \alpha -1 ) \right) = 2f \beta + k ( \alpha -1 ) \le 2f (\beta + \alpha -1).
\end{align*}
We can finally write:
\begin{align*}
d_{H-F}(s, t) & \le d_{H-F}(s, x^{B^*}_1) + d_{H-F}(x^{B^*}_1,t)\\
& \le d_{G-F}(s, x^{B^*}_1) + 2f(\beta + \alpha -1) + d_{G-F}(x^{B^*}_1,t) +\beta\\
& \le d_{G-F}(s,t) + 2f(\beta + \alpha -1) + \beta.
\end{align*}
\end{proof}

It turns out that for the special case of $f=1$, we can provide a better upper bound on the distortion provided by $H$. However, this refinement currently does not imply any better additive spanner w.r.t. the ones described in Section \ref{sec:augmenting_clustering_spanners}.

\begin{lemma}
	\label{lem:onefault}
	Let $A$ be a $\beta$-additive spanner of $G$, and let \fab{$M$} be an $\alpha$-multiplicative EFT-spanner of $G$. The graph $H=(V(G), E(A) \cup E(\fab{M}))$ is a $(2\beta+\alpha-1)$-additive EFT spanner of $G$. 
\end{lemma}
\begin{proof}
	Fix any two vertices $s,t \in \fab{V(G)}$ and a failed edge $\fab{e} \in E(G)$. If $d_A(s,t) = d_{A-\fab{e}}(s,t)$, then we are done as $d_{H-\fab{e}}(s,t) \le d_{A-\fab{e}}(s,t) = d_{A}(s,t) \le d_G(s,t) + \beta \le d_{G-\fab{e}}(s,t) + \beta$.
	
	Otherwise, choose two vertices $z$ and $z^\prime$ in $\pi_{G-\fab{e}}(s,t)$ according to Lemma~\ref{lemma:spanner_union_one_fault_structure}. Notice that $(z, z^\prime)$ belongs to $G-\fab{e}$, hence $d_{\fab{M}-\fab{e}}(z, z^\prime) \le \alpha \cdot d_{G-\fab{e}}(z, z^\prime) = \alpha$. We have:
	\begin{align*}
		d_{H-\fab{e}}(s,t) & \le d_{A}(s,z) + d_{\fab{M}-\fab{e}}(z,z^\prime) + d_{A}(z^\prime, t) 
		 \le  d_G(s,z) + \beta + \alpha + d_G(z^\prime, t) + \beta \\
		& \le  d_{G-\fab{e}}(s,z) + (d_{G-\fab{e}}(z,z^\prime) -1 ) + d_{G-\fab{e}}(z^\prime, t) +  2\beta  + \alpha \\
		& = d_{G-\fab{e}}(s,t) + 2\beta + \alpha -1
	\end{align*}
\end{proof}

Theorem \ref{thm:multifault} follows immediately from Lemmas \ref{lem:multifault} and \ref{lem:onefault}.

\end{document}